\documentclass{amsart}
\usepackage[cp1251]{inputenc}
\usepackage[english]{babel}
\usepackage{amsmath}
\usepackage{amssymb}
\usepackage{amsfonts}
\newtheorem{thm}{Theorem}
\begin{document}

\thispagestyle{empty}

 \title[Solving the heat equation in piecewise-homogeneous anisotropic media ]{ Solving the heat equation in piecewise-homogeneous anisotropic media using the multidimensional  Fourier transforms
 }%

\author{Yaremko O.E.}%Указываем авторов

\maketitle {\small

\begin{quote}
\address{Oleg Emanuilovich Yaremko,
\newline\hphantom{iii}Faculty of physics and mathematics,Penza State University,% Место работы
\newline\hphantom{iii}str. Krasnaya, 40, % Адрес (улица, дом, строение и т.п.)
\newline\hphantom{iii} 440038, Penza, Russia,}%  Адрес (почтовый индекс, город, страна)
\email{yaremki@mail.ru}% Ваш электронный адрес для переписки
\medskip

\noindent{\bf Abstract. } Multidimensional integral transformations with non-separated variables for problems with discontinuous coefficients are constructed in this work. The coefficient discontinuities focused on the  of parallel hyperplanes. In this work explicit formulas for the kernels in the case of ideal coupling conditions are obtained; the basic identity of the integral transform is proved; technique of integral transforms is developed
\medskip

\noindent{\bf Keywords:} integral transformations,  non-separated variables, coupling conditions
\end{quote} }
	MSCS80A20  	Heat and mass transfer, heat flow; 		MSCS42A38, MSCS42B10   	Fourier and Fourier-Stieltjes transforms and other transforms of Fourier type.

\section{Introduction}
Integral transforms arise in a natural way through the principle of linear superposition in constructing integral representations of solutions of linear differential equations.
First note that the structure of integral transforms with
the relevant variables are determined by the type of differential equation and the kind of media in which the problem is considered. Therefore  decision of integral transforms  are the problem for
mathematical physics piecewise-homogeneous (heterogeneous) media. It is clear this method is an effective for obtaining
the exact solution of boundary-value problems for  piecewise-homogeneous structures mathematical physics.
Integral transforms and their applications are appeared in the mathematical monographs of Uflyand Y.S. [1,2],
Lenuk M.P. [3,4]; Nayda L.S. [4] , Protsenko V.S. [5], etc.
The author together with I.I.Bavrin has proposed integral transforms with non-separate variables for solving multidimensional problems  in the work \cite{yar}.

Let $V$ from $R^{n+1}$  be the half-space
\[
V=\left\{ {\left( {y_1 ,...,y_n ,x} \right)\in R^{n+1}:x>0}
\right\},
\]
then solution of the Dirichlet's problem for the half-space is expressed by Poisson formula takes the form:
\cite{bes}

\[
u(x,y)=\Gamma \left( {\frac{n+1}{2}} \right)\pi
^{-\frac{n+1}{2}}\int\limits_{y=0} {\frac{x}{\left[ (y-\eta )^2+x^2 \right]^{\frac{n+1}{2}}}f(\eta )d\eta } .
\]
Obviously Poisson's kernel  is the form of integral Laplace transform and therefore expansion of the function $f(y)$ for the eigenfunctions of the Laplace operator $\Delta$ is obtained from the  reproduce properties of the Poisson kernel:
$$
f(y)=\mathop {\lim }\limits_{\tau \to 0} \int\limits_0^\infty {\lambda
^{\frac{n}{2}}e^{-\lambda \tau }} \left(\frac{1}{\left( {\sqrt {2\pi }
} \right)^n}\int\limits_{R^n} {\frac{J_{\frac{n-2}{2}} \left( {\lambda
\left| {y-\eta } \right|} \right)}{\left| {y-\eta }
\right|^{^{\frac{n-2}{2}}}}} f\left( \eta \right)d\eta\right)d\lambda ,
$$
here $J_{\nu}$ is Bessel's function of order $\nu$ \cite{bes}.
 We may assume that integral transforms with non- separate variables are defined as follows \cite{yar} on the basis of this expansion:\\
direct integral Fourier transform has the form

\begin{equation} F\left[ f \right]\left( {y,\lambda } \right)=\frac{1}{\left( {\sqrt {2\pi }
} \right)^n}\int\limits_{R^n} {\frac{J_{\frac{n-2}{2}} \left( {\lambda
\left| {y-\eta } \right|} \right)}{\left| {y-\eta }
\right|^{^{\frac{n-2}{2}}}}} f\left( \eta \right)d\eta \equiv \hat {f}\left(
{y,\lambda } \right),
\end{equation}
inverse Fourier integral transform has the form
\begin{equation}
F^{-1} [\hat {f}](y)=
\mathop {\lim }\limits_{\tau \to 0} \int\limits_0^\infty {\lambda
^{\frac{n}{2}}e^{-\lambda \tau }} \hat {f}(y;\lambda )d\lambda \equiv f(y).
\end{equation}

In our case the construction of multi-dimensional analogues for integral transforms (1)-(2) with discontinuous coefficients is the purpose of this research.
\section{One-dimensional integral transforms with discontinuous coefficients}
In this paper integral transforms with discontinuous coefficients are constructed in accordance with  author's work \cite{10}. Let $\varphi \left( {x,\lambda } \right)$ and $\varphi ^\ast \left(
{x,\lambda } \right)$ be eigenfunctions of primal and dual problems
Sturm-Liouville for Fourier operator on sectionally  homogeneous axis $ I_n $,
\[
I_n =\left\{ {x:\;x\in \mathop
U\limits_{j=1}^{n+1} \left( {l_{j-1} ,l_j } \right),\;\,l_0 =-\infty
,\;\,l_{n+1} =\infty ,\;\,l_j <l_{j+1} ,\;\,j=\overline {1,n} } \right\}.
\]
Let us remark that eigenfunction $\varphi \left( {x,\lambda } \right)$,
\[
\varphi \left( {x,\lambda } \right)=\sum\nolimits_{k=2}^n {\theta \left(
{x-l_{k-1} } \right)\,\theta \left( {l_k -x} \right)\,\varphi _k \left(
{x,\lambda } \right)+}
\]
\[
+\,\theta \left( {l_1 -x} \right)\,\varphi _1 \left( {x,\lambda }
\right)+\theta \left( {x-l_n } \right)\,\varphi _{n+1} \left( {x,\lambda }
\right)
\]
is the solution of separated differential equations system
\[
\left( {a_m^2 \frac{d^2}{dx^2}+{\kern 1pt}\lambda ^2} \right)\,\varphi _m
\left( {x,\lambda } \right)=0,\;\;x\in \left( {l_m ,l_{m+1} } \right);\quad
m=1,...,n+1,
\]
by the coupling conditions
\[
\left[ {\alpha _{m1}^k \frac{d}{dx}+\beta _{m1}^k } \right]\varphi _k
=\left[ {\alpha _{m2}^k \frac{d}{dx}+\beta _{m2}^k } \right]\varphi _{k+1}
,
\]
\[
x=l_k ,\;\;k=1,...,n;\;\;m=1,2,
\]
on the boundary conditions
\[
\left. {\varphi _1 } \right|_{x=-\infty } =0\,,\;\,\left. {\;\varphi _{n+1}
} \right|_{x=\infty } =0.
\]
Similarly, the eigenfunction $\varphi ^\ast \left(
{x,\lambda } \right)$,
\[
\varphi ^\ast \left( {\xi ,\lambda } \right)=\sum\nolimits_{k=2}^n {\theta
\left( {\xi -l_{k-1} } \right)\,\theta \left( {l_k -\xi } \right)\,\varphi
_k^\ast \left( {\xi ,\lambda } \right)\,+}
\]
\[
+\theta \left( {l_1 -\xi } \right)\,\varphi _1^\ast \left( {\xi ,\lambda }
\right)+\theta \left( {\xi -l_n } \right)\,\varphi _{n+1}^\ast \left( {\xi
,\lambda } \right)
\]
is the solution  of separate differential equations system
\[
\left( {a_m^2 \frac{d^2}{dx^2}+{\kern 1pt}\lambda ^2} \right)\,\varphi
_m^\ast \left( {x,\lambda } \right)=0,\;\;x\in \left( {l_m ,l_{m+1} }
\right);\quad m=1,...,n+1,
\]
by the coupling conditions
\[
\frac{1}{\Delta _{1,k} }\left[ {\alpha _{m1}^k \frac{d}{dx}+\beta _{m1}^k }
\right]\varphi _k^\ast =\frac{1}{\Delta _{2,k} }\left[ {\alpha _{m2}^k
\frac{d}{dx}+\beta _{m2}^k } \right]\varphi _{k+1}^\ast ,
\quad
x=l_k ,\;\;
\]
where
\[
\Delta _{i,k} =\det \left( {{\begin{array}{*{20}c}
 {\alpha _{1i}^k } \hfill & {\beta _{1i}^k } \hfill \\
 {\alpha _{2i}^k } \hfill & {\beta _{2i}^k } \hfill \\
\end{array} }} \right)
k=1,...,n;\;\;\quad i,m=1,2,
\]
on the boundary conditions
\[
\left. {\varphi _1 } \right|_{x=-\infty } =0\,,\;\,\left. {\;\varphi _{n+1}
} \right|_{x=\infty } =0.
\]
Further normalization
eigenfunctions is accepted by the following:
\[\varphi _{n+1} \left( {x,\lambda }
\right)=e^{ia_{n+1}^{-1} x\lambda }. \quad \varphi _{n+\mbox{1}}^\ast \left(
{x,\lambda } \right)=e^{-ia_{n+1}^{-1} x\lambda }.\]
Let direct $F_{n} $ and inverse $F_{n}^{-1} $ Fourier transforms on the Cartesian axis with $ n $
division points be defined by the rules in [10] :
\begin{equation}
F_{n} \left[ f \right]\,\left( \lambda \right)=\sum\limits_{m=0}^{n+1}
{\int\limits_{l_m-1 }^{l_{m} } \; } u_{m}^\ast \left( {\xi ,\lambda }
\right)\,f_{m} \left( \xi \right)d\xi \equiv \hat {f}\left( \lambda
\right),
\end{equation}
\begin{equation}
f_k \left( x \right)=\frac{1}{\pi i}\int\limits_0^\infty {u_k \left(
{x,\lambda } \right)\hat {f}\left( \lambda \right)\lambda d\lambda .}
\end{equation}

\section{Heat conduction in piecewise-homogeneous anisotropic media}

 Here the method of delta- functions \cite{len} is
the foundation for integral transforms. Namely kernels of the integral solutions mixed Cauchy's problem are the delta- functions.

This means that we fined the solution
of the separated matrix systems  $ (n +1) $ parabolic equations:
\begin{equation}
\left( {\frac{\partial }{\partial t}-A_j^2 \frac{\partial ^2}{\partial
x^2}-\Delta _y } \right)U_j \left( {t,x,y} \right)=0,\;\;\left( {t,x,y}
\right)\in D_+ \times R^m,\;\,j=\overline {1,n+1}
\end{equation}
bounded on the set $D \times R^m,D^+=\left( {0,\infty } \right)\times I_n ,\;\,$

where
\[
I_n =\left\{ {x:\;x\in \mathop
U\limits_{j=1}^{n+1} \left( {l_{j-1} ,l_j } \right),\;\,l_0 =-\infty
,\;\,l_{n+1} =\infty ,\;\,l_j <l_{j+1} ,\;\,j=\overline {1,n} } \right\}
\]
\[
\Delta _y =\frac{\partial ^2}{\partial y_1^2 }+\ldots +\frac{\partial
^2}{\partial y_m^2 },
\]

$A_j =\left( {a_{kl}^j } \right)$- positive definite  matrix $r\times r$,

to initial data
\begin{equation}
\left. {U_j \left( {t,x,y} \right)\,} \right|_{t=0} =g_j \left( {x,y}
\right),\;\,x\in I_n ,y\in R^m
\end{equation}

by edge condition
\begin{equation}
\left. {U_1 } \right|_{x=-\infty } =0\,,\;\,\left. {U_{n+1} }
\right|_{x=\infty } =0
\end{equation}

by the coupling condition

\begin{equation}
\left[ {\alpha _{m1}^k \frac{\partial }{\partial x}+\beta _{m1}^k }
\right]U_k =\left[ {\alpha _{m2}^k \frac{\partial }{\partial x}+\beta
_{m2}^k } \right]U_{k+1} ,
\end{equation}
\[
x=l_k ,\;\;k=\overline {1,n} ;\;\;m=\overline {1,\,2} ,
\]

here $U_{j}(t,x,y)$ is unknown vector- function, $g_{j}(x,y)$ given
vector- function, $\alpha _{mi}^k ,\;\,\beta _{mi}^k ,\\\;\,\gamma _{mi}^k
,\;\,\delta _{mi}^k $  an  matrix's $r\times r$.

Fourier integral with discontinuous coefficients
section 2 and the Fourier integral with non-separated variable
(1) - (2) give the idea of solving the problem (3) - (6). This solution takes the form:

\[
U_k \left( {t,x,y} \right)=-\frac{1}{\pi i}\frac{1}{\left( {\sqrt {2\pi } }
\right)^m}\int\limits_{R^m} {\sum\limits_{j=1}^{n+1} {\int\limits_{l_{j-1}
}^{l_{j} } {\mathop {\lim }\limits_{\tau \to 0} } } } \left(
{\int\limits_0^\infty {\frac{J_{\frac{m-2}{2}} \left( {\lambda \left|
{y-\eta } \right|} \right)}{\left| {y-\eta }
\right|^{^{\frac{m-2}{2}}}}e^{-\lambda \tau }\lambda ^{\frac{m}{2}}d\lambda
} } \right.\cdot\]
\begin{equation}
\cdot \int\limits_{-\infty }^\infty {e^{-\beta ^2t}\varphi _k \left( {x,\beta }
\right)} \varphi _j^\ast \left( {\xi ,\beta } \right)d\beta \left.
{{\begin{array}{*{20}c}
 \hfill \\
 \hfill \\
\end{array} }} \right)f_j \left( {\xi ,\eta } \right)d\xi d\eta ,
\quad
k=\overline {1,n+1} ,
\end{equation}
where $\varphi _k \left( {x,\beta } \right),\varphi _j^\ast \left( {\xi ,\beta }
\right)$ are the eigenfunctions of the direct and dual Sturm --Liouville problems,
respectively.

Suppose iterated integral

\[
\int\limits_0^\infty {\frac{J_{\frac{m-2}{2}} \left( {\lambda \left| {y-\eta
} \right|} \right)}{\left| {y-\eta } \right|^{^{\frac{m-2}{2}}}}e^{-\lambda
\tau }\lambda ^{\frac{m}{2}}d\lambda \int\limits_{-\infty }^\infty
{e^{-\beta ^2t}} } \varphi _k \left( {x,\beta } \right)\varphi _j^\ast \left( {\xi
,\beta } \right)d\beta
\]
is considered as a double then pass to the polar coordinates $$\lambda =\rho
\sin \varphi ,\beta =\rho \cos \varphi ;0\le \rho <\infty ,0\le \varphi \le \pi .$$
Then we prove:

\[
 \int\limits_0^\infty {\rho ^{\frac{m}{2}}\rho d\rho
{\int\limits_0^\pi {e^{-\rho ^2t\cos ^2\alpha }} } .\sin^{\frac{m}{2}}\alpha \frac{J_{\frac{m-2}{2}} \left( {\rho \sin \alpha \left|
{y-\eta } \right|} \right)}{\left| {y-\eta }
\right|^{^{\frac{m-2}{2}}}}e^{-\rho \tau \sin \alpha }} \cdot \\
\]
\[
 \cdot \varphi _k ( x,\rho \cos \alpha )\varphi _j^\ast
( \xi ,\rho \cos \alpha )d\alpha.
\]

If perform to the limit as $\tau\to 0 $ in the formula (9), then we obtain
\[
u_k \left( {t,x,y} \right)=-\frac{1}{\pi i}\frac{1}{\left( {\sqrt {2\pi } }
\right)^m}\int\limits_{R^m} {\sum\limits_{j=1}^{n+1} {\int\limits_{l_{j-1}
}^{l_{j} } {\int\limits_0^\infty e^{-\lambda^{2}  t}{\rho ^{\frac{m}{2}}\rho d\rho \left(
{\int\limits_0^\pi {si} } \right.n^{\frac{m}{2}}\alpha
\frac{J_{\frac{m-2}{2}} \left( {\rho \sin \alpha \left| {y-\eta } \right|}
\right)}{\left| {y-\eta } \right|^{^{\frac{m-2}{2}}}}} } } } \cdot
\]

\begin{equation}
\cdot \varphi _k \left( {x,\rho \cos \alpha } \right)\varphi _j^\ast \left( {\xi
,\rho \cos \alpha } \right)d\alpha \left. {{\begin{array}{*{20}c}
 \hfill \\
 \hfill \\
\end{array} }} \right)f_j \left( {\xi ,\eta } \right)d\xi d\eta .
\end{equation}
Here we use the notation
\[
\varphi _{k,j} \equiv\varphi _{k,j} \left( {\rho ,x,\xi ,\left| {y-\eta } \right|}
\right)=\int\limits_0^\pi {si} n^{\frac{m}{2}}\alpha \frac{J_{\frac{m-2}{2}}
\left( {\rho \sin \alpha \left| {y-\eta } \right|} \right)}{\left| {y-\eta }
\right|^{^{\frac{m-2}{2}}}}\cdot
\]
\[
\cdot \varphi _k \left( {x,\rho \cos \alpha } \right)\varphi _j^\ast \left( {\xi
,\rho \cos \alpha } \right)d\alpha
\]
I t is clear formula (10) be written as
\begin{equation}
\label{eq8}
u_k \left( {t,x,y} \right)=\frac{1}{\pi }\int\limits_0^\infty e^{-\rho^{2}t}\rho
^{\frac{m+1}{2}}d\rho  {\frac{1}{( \sqrt 2\pi )^m}\int\limits_{R^m} {\sum\limits_{j=1}^{n+1} {\int\limits_{l_{j-1}
}^{l_{j} } {\varphi _{k,j} } } } f_j ({\xi ,\eta })d\xi d\eta }.
\end{equation}
\section{The multidimensional Fourier transforms with discontinuous on planes}
If perform to the limit as $t\to 0 $ in the formula (11), then we obtain
\begin{equation}
\label{eq8}
f_k \left( {x,y} \right)=\frac{1}{\pi }\int\limits_0^\infty \rho
^{\frac{m+1}{2}}d\rho  {\frac{1}{( \sqrt 2\pi )^m}\int\limits_{R^m} {\sum\limits_{j=1}^{n+1} {\int\limits_{l_{j-1}
}^{l_{j} } {\varphi _{k,j} } } } f_j ({\xi ,\eta })d\xi d\eta }.
\end{equation}
Note that direct and inverse
multidimensional Fourier transforms with discontinuous on planes
$x =l_k$  are determined due to the integral identity (12):
\begin{equation}
F_n \left[ f \right]\left( {x,y,\lambda } \right)=\frac{1}{\left( {\sqrt
{2\pi } } \right)^m}\int\limits_{R^m} {\sum\limits_{j=1}^{n+1}
{\int\limits_{l_{j-1} }^{l_{j} } } } \varphi _{k,j} \left( {\lambda ,x,\xi ,\left|
{y-\eta } \right|} \right)f_j \left( {\xi ,\eta } \right)d\xi d\eta
,
\end{equation}
\\
\begin{equation}
f(x,y)=
\int\limits_0^\infty {\lambda ^{\frac{m}{2}+1}}F_n \left[ f \right]\left( {x,y,\lambda } \right)d\lambda ,
\end{equation}
respectively.

Finally we prove the basic integral identity for differential
operator
\[
\begin{array}{l}
 B=\,\theta \left( {l_1 -x} t\right)\left( {A_1^2 \frac{d^2}{dx^2}+\Delta _y
} \right)+\sum\limits_{k=1}^n {\,\theta \left( {x-l_{k-1} } \right)\,\theta
\left( {l_k -x} \right)\left( {A_k^2 \frac{d^2}{dx^2}+\Delta _y }
\right)\,+} \\
 +\,\theta \left( {x-l_n } \right)\left( {A_{n+1}^2 \frac{d^2}{dx^2}+\Delta
_y } \right).\\
 \end{array}
\]

\begin{thm} Suppose twice continuously differentiable on
$D_+ \times R^m$ vector- function

\[
f(x,y)=\sum_{k=1}^{n+1}\theta \left( {x-l_{k-1} }
\right)\theta \left( {l_k -x }
\right)\,f_{k} \left( {x,y} \right)+
\]

\[
+\theta \left( {x-l_n }
\right)\,f_{n+1} \left( {x,y} \right),
\]
satisfies the conditions at infinity:\\

$$f_{n+1}(x,y),\frac{\partial f_{n+1}(x,y)}{\partial x}$$
 limits to zero as $x\rightarrow+\infty,y$ -is fixed,\\

$$f_{1}(x,y),\frac{\partial f_{1}(x,y)}{\partial x}$$
limits to zero as  $x\rightarrow-\infty,y$ -is fixed,\\

$$f_{i}(x,y),\frac{\partial f_{i}(x,y)}{\partial y_{j}}$$
limits to zero as $y_{j}\rightarrow\pm\infty;x,y_{1},y_{2},...,y_{j-1},y_{j+1},...,y_{m}$ -are fixed,

the coupling conditions (8) are performed

then hold true:

\[
F_n \left[ {B\left( f \right)} \right]=-\lambda ^2F_n \left[ f \right].
\]
\end{thm}
\begin{proof}
Let us twice integrate   by parts with respect to each of the variables in the left part. Further outside the integral terms are disappeared by the conditions at infinity, connection summands  are disappeared by the coupling conditions. Therefore the operator $B$ is placed as the kernel:
\[F_n \left[ B(f) \right]\left( {x,y,\lambda } \right)=\frac{1}{\left( {\sqrt
{2\pi } } \right)^m}\int\limits_{R^m} {\sum\limits_{j=1}^{n+1}
{\int\limits_{l_{j-1} }^{l_{j} } } } B_{j}\left[\varphi _{k,j} \left( {\lambda ,x,\xi ,\left|
{y-\eta } \right|} \right)\right]f_j \left( {\xi ,\eta } \right)d\xi d\eta
.
\]
Hence the equality  $B_{j}[\varphi_{k,j}]=-\lambda^{2}\varphi_{k,j}$  is proved. 	
We shall prove that |
\[
B_{j}[\varphi _{k,j}] =\int\limits_0^\pi {si} n^{\frac{m}{2}}\alpha \Delta_{\eta}\left(\frac{J_{\frac{m-2}{2}}
\left( {\rho \sin \alpha \left| {y-\eta } \right|} \right)}{\left| {y-\eta }
\right|^{^{\frac{m-2}{2}}}}\right)\cdot
\]
\[
\cdot \varphi _k \left( {x,\rho \cos \alpha } \right)\left(\varphi _j^\ast \left( {\xi
,\rho \cos \alpha } \right)\right)d\alpha+
\]
\[
 +\int\limits_0^\pi {si} n^{\frac{m}{2}}\alpha \left(\frac{J_{\frac{m-2}{2}}
\left( {\rho \sin \alpha \left| {y-\eta } \right|} \right)}{\left| {y-\eta }
\right|^{^{\frac{m-2}{2}}}}\right)\cdot
\]
\[
\cdot \varphi _k \left( {x,\rho \cos \alpha } \right)a_{j}^{2}\frac{\partial^{2}}{\partial\xi^{2}}\left(\varphi _j^\ast \left( {\xi
,\rho \cos \alpha } \right)\right)d\alpha=
\]
\[
=-\rho^{2} \sin^{2} \alpha\int\limits_0^\pi {si} n^{\frac{m}{2}}\alpha \left(\frac{J_{\frac{m-2}{2}}
\left( {\rho \sin \alpha \left| {y-\eta } \right|} \right)}{\left| {y-\eta }
\right|^{^{\frac{m-2}{2}}}}\right)\cdot
\]
\[
\cdot \varphi _k \left( {x,\rho \cos \alpha } \right)\left(\varphi _j^\ast \left( {\xi
,\rho \cos \alpha } \right)\right)d\alpha+
\]
\[
 -\rho^{2} \cos^{2} \alpha\int\limits_0^\pi {si} n^{\frac{m}{2}}\alpha \left(\frac{J_{\frac{m-2}{2}}
\left( {\rho \sin \alpha \left| {y-\eta } \right|} \right)}{\left| {y-\eta }
\right|^{^{\frac{m-2}{2}}}}\right)\cdot
\]
\[
\cdot \varphi _k \left( {x,\rho \cos \alpha } \right)\left(\varphi _j^\ast \left( {\xi
,\rho \cos \alpha } \right)\right)d\alpha=-\rho^{2}\varphi _{k,j}.
\]
We take into consideration that in this proof $\varphi _j^\ast \left( {\xi
,\rho \cos \alpha } \right)$ are  dual Sturm-Liouville problems eigenfunctions.
$$
\Delta_{\eta}\left(\frac{J_{\frac{m-2}{2}}
\left( {\rho \sin \alpha \left| {y-\eta } \right|} \right)}{\left| {y-\eta }
\right|^{^{\frac{m-2}{2}}}}\right)=-\rho^{2} \sin^{2} \alpha\left(\frac{J_{\frac{m-2}{2}}
\left( {\rho \sin \alpha \left| {y-\eta } \right|} \right)}{\left| {y-\eta }
\right|^{^{\frac{m-2}{2}}}}\right),
$$
By basis identity \cite{9} we conclude
$$
\frac{\rho^{\frac{m}{2}}J_{\frac{m-2}{2}}
\left( {\rho \left| {y } \right|} \right)}{\left| {y }
\right|^{^{\frac{m-2}{2}}}}=\frac{1}{(2\pi)^{\frac{m}{2}}}\int_{S_{\rho}}e^{i<y,\xi>}dS_{\rho}.
$$
\end{proof}
This completes the proof.

Specifically the formulas for the direct and inverse Fourier transforms with
non- separated  variables are significantly simplified in the case  of ideal coupling conditions on one surface. This case is the most
distributed in engineering practice. 	
As an example the scalar case is considered.
Suppose the ideal coupling conditions are in the plane $x=0$
\[
 \varphi_1 \left( {x,y} \right)=\varphi_2 \left( {x,y} \right),x=0,y\in R^m;
 \]
 \[
 {\varphi}'_{1x} \left( {x,y} \right)=\nu {\varphi}'_{2x} \left( {x,y} \right),x=0,y\in
R^m;\nu =\frac{\lambda _2 }{\lambda _1 }
\]
then the analytical expressions for the one-dimensional components
eigenfunctions are in \cite {len}:
\[
 \varphi _1 \left( {x,\lambda } \right)=\left( {\cos \lambda \frac{x}{a_1
}+i\frac{1}{\sqrt {\delta _0 } }\sin \lambda \frac{x}{a_1 }} \right)\left(
{1+\delta _0 } \right);
\]
\[
 \varphi _2 \left( {x,\lambda } \right)=\left( {\cos \lambda \frac{x}{a_2
}+i\sqrt {\delta _0 } \sin \lambda \frac{x}{a_2 }} \right)\left( {1+\delta
_0 } \right);
 \]
\[
\varphi _k^\ast \left( {x,\lambda } \right)=r_k \overline {\varphi _k \left(
{x,\lambda } \right)} ,k=1,2,r_1 =\frac{a_2 }{\nu _0 a_1^2 },r_2
=\frac{1}{a_2 },\delta _0 =\frac{a_2 }{\nu _0 a_1 }.
\]
It is clear that the expressions for the multidimensional components of eigenfunctions with
non- separate variables $\varphi_ {kj}$ have the form:
\[
\varphi _{11} =\frac{1+\delta _0 }{a_1 }\frac{J_{\frac{m-1}{2}} \left( {\rho
\sqrt {\frac{\left( {x-\xi } \right)^2}{a_1^2 }+\left| {y-\eta } \right|^2}
} \right)}{\left( {\frac{\left( {x-\xi } \right)^2}{a_1^2 }+\left| {y-\eta }
\right|^2} \right)^{^{\frac{m-1}{2}}}}-\frac{1-\delta _0 }{a_1
}\frac{J_{\frac{m-1}{2}} \left( {\rho \sqrt {\frac{\left( {x+\xi }
\right)^2}{a_1^2 }+\left| {y-\eta } \right|^2} } \right)}{\left(
{\frac{\left( {x+\xi } \right)^2}{a_1^2 }+\left| {y-\eta } \right|^2}
\right)^{^{\frac{m-1}{2}}}},
\]
\[
\varphi _{12} =\frac{1+\delta _0 }{a_2 \sqrt {\delta _0 }
}\frac{J_{\frac{m-1}{2}} \left( {\rho \sqrt {\left( {\frac{x}{a_2
}-\frac{\xi }{a_1 }} \right)^2+\left| {y-\eta } \right|^2} } \right)}{\left(
{\left( {\frac{x}{a_2 }-\frac{\xi }{a_1 }} \right)^2+\left| {y-\eta }
\right|^2} \right)^{^{\frac{m-1}{2}}}}+
\]
\[
+\frac{1-\delta _0 }{a_2 \sqrt {\delta
_0 } }\frac{J_{\frac{m-1}{2}} \left( {\rho \sqrt {\left( {\frac{x}{a_2
}+\frac{\xi }{a_1 }} \right)^2+\left| {y-\eta } \right|^2} } \right)}{\left(
{\left( {\frac{x}{a_2 }+\frac{\xi }{a_1 }} \right)^2+\left| {y-\eta }
\right|^2} \right)^{^{\frac{m-1}{2}}}},
\]
\[
\varphi _{21} =\sqrt {\delta _0 } \frac{1+\delta _0 }{a_1
}\frac{J_{\frac{m-1}{2}} \left( {\rho \sqrt {\left( {\frac{x}{a_1
}-\frac{\xi }{a_2 }} \right)^2+\left| {y-\eta } \right|^2} } \right)}{\left(
{\left( {\frac{x}{a_1 }-\frac{\xi }{a_2 }} \right)^2+\left| {y-\eta }
\right|^2} \right)^{^{\frac{m-1}{2}}}}+
\]
\[
+\sqrt {\delta _0 } \frac{1-\delta _0
}{a_1 }\frac{J_{\frac{m-1}{2}} \left( {\rho \sqrt {\left( {\frac{x}{a_1
}+\frac{\xi }{a_2 }} \right)^2+\left| {y-\eta } \right|^2} } \right)}{\left(
{\left( {\frac{x}{a_1 }+\frac{\xi }{a_2 }} \right)^2+\left| {y-\eta }
\right|^2} \right)^{^{\frac{m-1}{2}}}},
\]
\[
\varphi _{22} =\frac{1+\delta _0 }{a_2 \delta _0 }\frac{J_{\frac{m-1}{2}}
\left( {\rho \sqrt {\frac{\left( {x-\xi } \right)^2}{a_2^2 }+\left| {y-\eta
} \right|^2} } \right)}{\left( {\frac{\left( {x-\xi } \right)^2}{a_2^2
}+\left| {y-\eta } \right|^2} \right)^{^{\frac{m-1}{2}}}}-\frac{1-\delta _0
}{a_2 \delta _0 }\frac{J_{\frac{m-1}{2}} \left( {\rho \sqrt {\frac{\left(
{x+\xi } \right)^2}{a_2^2 }+\left| {y-\eta } \right|^2} } \right)}{\left(
{\frac{\left( {x+\xi } \right)^2}{a_2^2 }+\left| {y-\eta } \right|^2}
\right)^{^{\frac{m-1}{2}}}}.
\]
This shows that integral transforms by formulas (12) - (13) are constructed.
\section{Conclusion}
Let us remark that integral transforms (12) - (13) are used in solving problems of mathematical physics by the standard algorithm. If we find the solution in the images then return to the original. We stress if one spectral parameter involved in the final formula then the practical profit  is achieved. At the same time the integral transforms with separate variables contain $m$ parameters.

\end{document}